\documentclass[reprint,aps,pra]{revtex4-2} 
\usepackage{acro}
\usepackage{xcolor}
\usepackage{dsfont}
\usepackage{amsmath}
\usepackage{amssymb}
\usepackage{amsthm}
\usepackage{colortbl}
\usepackage{algorithm}
\usepackage{quantikz}
\usepackage{hyperref}
\usepackage{orcidlink}
\usepackage[caption=false]{subfig}
\usepackage{enumitem}
\usepackage[nameinlink,capitalize]{cleveref}

\newcommand{\true}{\relax\ifmmode\mathrm{true}\else\textit{true}\fi}
\newcommand{\false}{\relax\ifmmode\mathrm{false}\else\textit{false}\fi}

\newtheorem{theorem}{Theorem}
\newtheorem{definition}{Definition}
\newtheorem{lemma}{Lemma}
\DeclareMathOperator{\tr}{tr}
\DeclareMathOperator*{\ubigvee}{\underline{\bigvee}}

\definecolor{mygreen}{rgb}{0.45098, 0.635294, 0.215686}
\definecolor{myblue}{rgb}{0, 0.4, 0.552941}
\definecolor{myorange}{rgb}{0.878431, 0.690196, 0.180392}
\definecolor{mygrey}{rgb}{0.407843, 0.407843, 0.403922}

\DeclareAcronym{NP}{
short = NP,
long = nondeterministic polynomial time
}
\DeclareAcronym{PVM}{
short = PVM,
long = projector-valued measure
}
\DeclareAcronym{FAPP}{
short = FAPP,
long = for all practical purposes
}
\DeclareAcronym{POVM}{
short = POVM,
long = positive operator-valued measure
}
\DeclareAcronym{BQP}{
short = BQP,
long = bounded-error quantum polynomial time
}
\DeclareAcronym{CSS}{
short = CSS,
long = Calderbank-Shor-Steane
}
\DeclareAcronym{CNF}{
short = CNF,
long = conjunctive normal form
}
\DeclareAcronym{XOR}{
short = XOR,
long = exclusive or
}
\DeclareAcronym{SSR}{
short = SSR,
long = superselection rule
}

\begin{document}

\preprint{APS/123-QED}

\title{Unobservables and Decoherence from Complexity}%

\author{Michael Epping\orcidlink{0000-0003-0950-6801}}%
\email{Michael.Epping@dlr.de}
\author{Jochen Szangolies\orcidlink{0000-0003-1141-3225}}%
\email{Jochen.Szangolies@dlr.de}
\affiliation{Institute of Software Technology, German Aerospace Center (DLR), Germany.}%

\date{\today}

\begin{abstract}
The interface between the quantum and the classical is an intriguing and, at times, hotly contested subject of ongoing research. 
The quantum regime is characterized by interference, made possible by the superposition principle, while such phenomena are absent in macroscopic, everyday experience. 
Here, we investigate the link of this absence (or, as we will argue, unobservability) to computational complexity. 
We show how the assumption that quantum systems cannot solve NP-complete problems efficiently implies that certain formally valid quantum measurements on finite-dimensional systems are unperformable.
We study several consequences of this restriction.
First, Pauli matrices in an inconveniently transformed basis are a simple example of unobservables.
Furthermore, some quantum states are not connected by any physically realizable time evolution.
Finally there are quantum states whose coherence cannot be observed, i.e. superpositions of pure quantum states which are indistinguishable from mixtures. 
We discuss the connection of this phenomenon to the presence of superselection sectors.
Our results suggest that the apparent classicality of macroscopic systems may be partly due to limitations on measurements and time evolutions imposed by computational complexity.
\end{abstract}

\maketitle

Individual quantum measurements can be viewed as questions posed by an observer and answered by nature.
More complex phenomena emerge with the shift from observers to active agents that adjust future actions based on measurement outcomes~\cite{fuchs2017notwithstanding}.
One interesting example is the undoing of measurements in case of undesired outcomes~\cite{wigner1995remarks,ueda1998logical,jordan2010uncollapsing,schindler2013undoing,zurek2018quantum}.
This potentially provides a powerful tool to the agent.
In this paper we demonstrate that in certain instances it is indeed too powerful to be compatible with the existence of computationally hard problems.
We construct a formally valid \ac{POVM} that would allow an efficient solution to the 3SAT problem. Thus, we propose that performing this measurement must itself be prohibitively complex, and thus it is not realized in nature.
Thus, while quantum computing is a key area where our results apply, our argument is independent of task-specific quantifiers of complexity such as circuit complexity, and also applies to (yet unknown) ``native gates'' as provided by nature.
For sufficiently large problem instances, such measurements are thus unreachable for any observer with finite resources. Hence, we term the associated quantities \emph{unobservables}.
In particular, we show that there exist states such that their superposition is operationally indistinguishable from a mixture.
Thus, a main consequence of unobservables is decoherence in the sense that coherence of certain states cannot be observed.
This opens up the possibility that complexity, rather than environmental interactions as in the more common approach~\cite{joos1985emergence, joos2013decoherence}, may be a source of decoherence. 

Prior work by Aaronson et al. has demonstrated a link between the complexity of measuring superpositions and that of unitary time evolution of states~\cite{aaronson2020hardness}.
Here, we link the observability of coherence directly to a computationally hard problem for the first time.

Quantum measurement is usually cast in terms of irreversibility. However, for present purposes, it is important to note that sometimes quantum measurements can indeed be undone.
One approach towards the quantum measurement problem(s)~\cite{fine1973two, maudlin1995three} is considering a global unitary time evolution that entangles system and environment.
With full control over the environment we could in principle undo the measurement.
Similar reasoning applies to `Wigner's Friend'-type scenarios~\cite{wigner1995remarks,zurek2018quantum}.
Other schemes for `undoing' measurements are based on weak measurements~\cite{ueda1998logical,jordan2010uncollapsing}.
Measurements can also be treated as errors that are corrected using quantum error correction codes~\cite{schindler2013undoing}.
The smallest possible example is the four-qubit \ac{CSS} code~\cite{CalderbankShor96, Steane1996} which encodes two logical qubits into four physical ones. It can correct a single measurement `error' at a known position, see Appendix~\ref{sec:fourqubitcode}. 

We proceed to sketch an algorithm that attempts to exploit the undoing of measurements by performing a follow-up measurement to `reset' unwanted outcomes (see Appendix~\ref{sec:algorithm} for details).
The algorithm is designed to determine the satisfiability of a propositional logic formula $F$ in \ac{CNF}. $F$ contains $N$ variables $v_1$, $v_2$, \dots, $v_N$ which appear in $M$ clauses $C_m$ of exactly three different literals each.
We denote literals $l_{mj} = \sigma_{mj} v_{q_{mj}}$, where $\sigma_{mj}=\pm 1$ and $q_{mj}\in\{1,2,\dots,N\}$ (a negative literal $-v_i$ corresponds to $\lnot v_i$). 
Then
\begin{equation}
F = \bigwedge_{m=1}^M \bigvee_{j=1}^3  l_{mj}.
\end{equation}
We associate a qubit with each of the $N$ variables, such that $\ket{0}_i$ and $\ket{1}_i$ correspond to the values \false{} and \true{} of the variable  $v_i$, respectively.
We say that a state fulfills a clause, if the state has no component violating it.
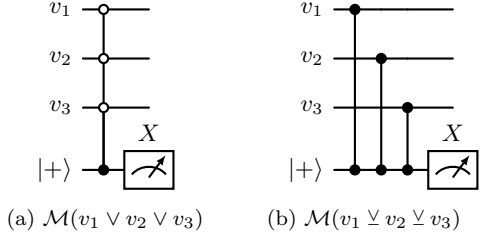
\begin{figure}[tbph]%
\centering%
\subfloat[$\mathcal{M}(v_1 \lor v_2 \lor v_3)$]{%
\hspace{0.3cm}%
\begin{quantikz}[column sep=2mm]%
\lstick{$v_1$} & \octrl{1} & \\%
\lstick{$v_2$} &\octrl{1} &\\%
\lstick{$v_3$} &\octrl{1} & \\%
\lstick{$\ket{+}$} &\ctrl{-1}& \meter{X}%
\end{quantikz}%
\hspace{0.3cm}%
\label{fig:measurementforclause}%
}%
\hspace{0.5cm}%
\subfloat[$\mathcal{M}(v_1\veebar v_2 \veebar v_3)$]{\hspace{0.5cm}%
\begin{quantikz}[column sep=2mm]%
\lstick{$v_1$} & \ctrl{3} & & &\\%
\lstick{$v_2$} & & \ctrl{2} & &\\%
\lstick{$v_3$} & & & \ctrl{1} & \\%
\lstick{$\ket{+}$} &\ctrl{0}&\ctrl{0}&\ctrl{0}& \meter{X}%
\end{quantikz}%
\label{fig:measurementforxorclause}%
}%
\caption{Circuit implementations for $\mathcal{M}$. We use open and filled circles if the gate is controlled by $\ket{0}$ or $\ket{1}$ corresponding to positive and negative literals, respectively.}%
\end{figure}%
The algorithm is based on the following quantum measurements.
In a slight abuse of notation we use the same symbols for the general case and the specific formula $F$ defined above.
Let $\Pi(C_m)$ and $\Pi(\lnot C_m)$ project onto the subspace of states fulfilling and violating clause $C_m$, respectively.
The corresponding projective measurement is
\begin{equation}
    \mathcal{M}(C_m) = \{\Pi(C_m),\; \Pi(\lnot C_m)\}. \label{eq:defmeasurement}
\end{equation}
A possible implementation of $\mathcal{M}(C_m)$ is shown in \cref{fig:measurementforclause}. 
Analogously we define $\mathcal{M}$ also for \ac{XOR}-clauses with a possible implementation shown in \cref{fig:measurementforxorclause}.
Finally, we define for a formula $F=\bigwedge_{i=1}^M C_i$ the \emph{sandwich}-\ac{POVM}
\begin{equation}
    \mathcal{N}(F) =\{P_x(F) Q_y P_x(F)| 0\leq x < 2^M, 0\leq y< 2^N\},
\end{equation}
where 
\begin{equation}
    P_x(F) = \prod_{i=1}^M\Pi((-1)^{x_i} C_i)
\end{equation}
can be considered a sequential measurement of the $\mathcal{M}(c_i)$ and
\begin{equation}
    Q_y = \bigotimes_{i=1}^N \frac{\mathds{1}+(-1)^{y_i} X}{2}
\end{equation}
is a measurement in the $X$-basis with outcomes $x=x_1,x_2,...,x_t$ and $y$, respectively. 
The post-measurement state of the sandwich-\ac{POVM} is obtained via the quantum channel with Kraus operators
\begin{equation}
\begin{aligned}
     M_{x,y} =&\sqrt{ P_x(F) Q_y P_x(F) } \\
     =& \sqrt{\frac{2^N}{\tr P_x(F)}} P_x(F) Q_y P_x(F).
\end{aligned}
\end{equation}
The last equation holds, because $P_x(F) Q_y P_x(F)$ is rank one. 
$\mathcal{N}$ undoes measurements of $\mathcal{M}$ on $P_x(F)\ket{+}^{\otimes N}$ up to irrelevant phases.
The right $P_x(F)$ fixes the outcome $x$, the $Q_y$ erases any unwanted projection introduced by $\mathcal{M}$, before the left $P_x(F)$ projects back to the support of $P_x(F)$.

Even though $P_x(F)$ and $Q_y$ are easy to implement, $\mathcal{N}$ may not be experimentally realizable for sufficiently large $N$, e.g. for hundreds of variables~\cite{VanGelder2010}.
In fact we will lead the following two main  assumptions to a contradiction:
\begin{description}
    \item[Assumption (INC)] \Ac{NP} $\not\subseteq$ \ac{BQP}~\cite{aaronson2005guest}, i.e. quantum systems do not allow to efficiently solve \ac{NP}-hard problems.
    \item[Assumption (OBS)] Every formally valid measurement can be experimentally realized, including $\mathcal{N}(F)$, for a formula $F$ such that its satisfiability problem is too hard to solve (with the finite resources available to any observer).
\end{description}

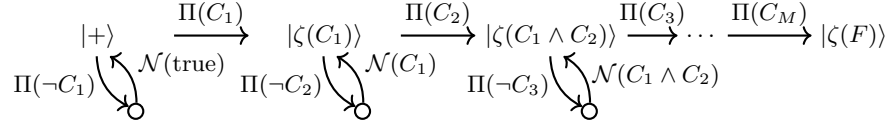
\begin{figure*}%
    \centering%
    \begin{tikzpicture}%
[lineDecorate/.style={->,thick},%
  ketDecorate/.style={minimum width=2cm},%
  nodeDecorate/.style={shape=circle,inner sep=2pt,draw,thick}]%
\node (b1) at (0.5,0) [nodeDecorate] {};%
\node (b2) at (3.5,0) [nodeDecorate] {};%
\node (b3) at (6.5,0) [nodeDecorate] {};%
\node (u1) at (0,1) [ketDecorate] {$\ket{+}$};%
\node (u2) at (3,1) [ketDecorate] {$\ket{\zeta(C_1)}$};%
\node (u3) at (6,1) [ketDecorate] {$\ket{\zeta(C_1\land C_2)}$};%
\node (u4) at (8,1) [] {\dots};%
\node (u5) at (10,1) [] {$\ket{\zeta(F)}$};%
\draw[lineDecorate] (u1) edge["$\Pi(C_1)$"] (u2);%
\draw[lineDecorate] (u2) edge["$\Pi(C_2)$"] (u3);%
\draw[lineDecorate] (u3) edge["$\Pi(C_3)$"] (u4);%
\draw[lineDecorate] (u4) edge["$\Pi(C_M)$"] (u5);%
\draw[lineDecorate,  bend right] (u1) edge node [midway, left] {$\Pi(\lnot C_1)$} (b1);%
\draw[lineDecorate,  bend right] (b1) edge[]  node [midway, right, yshift=1ex] {$\mathcal{N}(\true)$} (u1.300);%
\draw[lineDecorate,  bend right] (u2) edge node [midway, left] {$\Pi(\lnot C_2)$} (b2);%
\draw[lineDecorate,  bend right] (b2) edge[]  node [midway, right, yshift=1ex] {$\mathcal{N}(C_1)$} (u2.300);%
\draw[lineDecorate,  bend right] (u3) edge node [midway, left] {$\Pi(\lnot C_3)$} (b3);%
\draw[lineDecorate,  bend right] (b3) edge[]  node [midway, right] {$\mathcal{N}(C_1\land C_2)$} (u3.300);%
\end{tikzpicture}%
    \caption{The main idea of the hypothetical algorithm: Sequential measurements and reversing unwanted outcomes using a measurement $\mathcal{N}$ steers the state towards the satisfying assignments $\zeta(F)$ of a \ac{CNF} formula $F$ with clauses $C_m$ -- an \ac{NP}-hard problem. 
    ``Dilution'', see text, ensures the $\Pi(C_i)$ outcome is more likely than $\Pi(\lnot C_i)$.}%
    \label{fig:flowchart}%
\end{figure*}%
We exploit assumption (OBS) and the described undoing of the measurement $\mathcal{M}$ in the following algorithm to find a solution of $F$. 
We want to transform $\ket{+}^{\otimes N}$ into a solution of $F$ by sequentially measuring whether the state satisfies clause $C_m$, always undoing the unwanted answer `no'.
See \cref{fig:flowchart} for an overview of the main idea.
However, there is a technical difficulty that requires some additional effort:
The success probability of any of these steps might be prohibitively low at any $m$, see Appendix~\ref{sec:examplebadordering}.
To circumvent this, we ``dilute'' clauses to weaken their effect on the solution space, increasing the success probability for measuring the $\Pi(C_i)$ outcome.
By this we mean that we replace the formula $F$ by a `diluted' formula $\mathcal{D}(F)$ with $M'=6M$ clauses on $N'=N+M$ variables.
The $M$ new variables $v_{N+i}$ with $i=1,2,\dots,M$ encode whether the corresponding clauses $C_{i}$ are satisfied, i.e. $C_{i}\leftrightarrow v_{N+i}$.
The conversion into \ac{CNF} yields the first $4 M$ clauses.
The main trick, inspired by \emph{XORSample}~\cite{Gomes2006}, is to then add $M$ random \ac{XOR} clauses, that only check the parity of the number of satisfied clauses.
This roughly removes half of the non-solution assignments per random \ac{XOR} clause.
Finally we add $M$ clauses to ensure that all new variables are \true, such that $\mathcal{D}(F)$ is equivalent to $F$.
We give an explicit construction and prove the following properties in Appendix~\ref{app:dilution}:
\begin{enumerate}[nosep]
    \item[(D1)] $\mathcal{D}(F)$ is satisfiable if and only if $F$ is satisfiable. In this case, the solutions of $F$ and $\mathcal{D}(F)$ are identical on the first $N$ variables.
    \item[(D2)] If $F$ is satisfiable, then the probability to measure $\Pi(C_i')$ for each $C_i'\in\mathcal{D}(F)$ is greater than $\frac{1}{2}$ on average.
\end{enumerate}
The steps of this hypothetical quantum algorithm are shown in \cref{alg:solveFusingN}. 
A detailed example is given in Appendix~\ref{sec:example}.
We remark that the algorithm cannot be simulated efficiently, e.g. the required size of a classical memory to store the quantum state scales exponentially in the number of variables.
\begin{figure}
\begin{algorithm}[H]
\caption{Solve $F$ via $\mathcal{D}(F)$ using $\mathcal{M}$ and $\mathcal{N}$.}
\label{alg:solveFusingN}
\begin{enumerate}
    \item Initialize $N+M$ qubits in $\ket{+}^{\otimes (N+M)}$.  
    \item Initialize the established formula $e$ with \true{} (empty formula). From now on $e$ will track the clauses satisfied by the current state and protected by $\mathcal{N}(e)$.
    \item For every clause $C_i'$ in $\mathcal{D}(F)$:
    \begin{enumerate}
        \item Measure $\mathcal{M}(C_i')$. 
        \item If the outcome is $\Pi(C_i')$, then set $e=e\land C_i'$. 
        \item If the outcome is $\Pi(\lnot C_i')$, then undo the unwanted projection by measuring $\mathcal{N}(e)$. 
        \item Try a fixed number of times $R$ to obtain the outcome $\Pi(C_i')$. Redraw $C_i'$ if it is a random \ac{XOR} clause.
        \item Return ``unsatisfiable'' if all $R$ attempts failed. This is correct with probability at least $1-\frac{1}{2^R}$.
    \end{enumerate}
    \item Measure in the computational basis. The first $N$ qubits yield the solution $v^*$ that satisfies $F$.
\end{enumerate}
\end{algorithm}
\end{figure}

\cref{alg:solveFusingN} would solve an \ac{NP}-hard problem in polynomial time. 
There are good reasons to believe this to be physically impossible; indeed, Aaronson has proposed that the impossibility to implement such a process might become a restriction on physical possibility with a similar status as the second law of thermodynamics, and prove a useful guide in the search for new theories \cite{aaronson2005guest}. 
Thus, we propose that this implies that assumption (OBS) fails to hold. 
The implication is severe:
There is a problem size for which the measurement of $\mathcal{N}(F)$ cannot be carried out with the finite resources of any observer. 
Thus, for this problem size, $\mathcal{N}(F)$ is a simple example of a \ac{POVM} that has no counter-part in any physical system, 
even for comparatively modest system sizes.
Consequently, there are mathematically valid quantum measurements that are not physically realizable.
Hence, the associated quantities are unobservables in the above sense.
While one might arrive at this conclusion via different routes, for present purposes any unobservable is related to a computationally hard problem that scales superpolynomially in the input size and the given input size is large enough such that any observer does not have enough resources to measure it. 
Note that this conclusion is independent of the concrete physical realization of any experiment.

We only considered ideal operations, i.e. we did not prove that measurements arbitrarily close to $\mathcal{N}$ are also unobservable.
However, any measurement that is close to $\mathcal{N}$ will still allow to solve the same 3SAT problem, so the same reasoning applies to approximations of $\mathcal{N}$ as well.
We sketch this reasoning in Appendix~\ref{app:robustness}, by proving that any realization of a quantum algorithm that is $\frac{\epsilon}{L}$-close in diamond norm to the ideal implementation, where $L$ is the number of operations, gives outcome probabilities that are $\epsilon$-close to the ideal one.

In the remainder we give a simple example and look at interesting consequences of unobservables:
Unitaries that do not correspond to physical time evolutions and macroscopic but finite states whose coherence is unobservable.

We established that the sandwich-\ac{POVM} $\mathcal{N}(f)$ is unobservable. 
This unobservable \ac{POVM} is linked to unobservable \acp{PVM} via Naimark's dilation theorem. 
The isometry~\cite[p. 285]{alma990007249620206446}
\begin{equation}
    V = \sum_{x=0}^{2^N-1}\sum_{y=0}^{2^M-1} \ket{x, y} \bra{y}H^{\otimes N} P_x(F),
\end{equation}
fulfills
\begin{equation}
    V^\dagger\proj{x,y} V = P_x(F) Q_y P_x(F).
\end{equation}
One can extend $V$ to a unitary $U$, such that the \ac{POVM} is equivalent to applying $U$ followed by a canonical basis measurement.
We define the corresponding Hermitian operator
\begin{equation}
    \tilde{Z} = U^\dagger Z U \label{eq:Ztilde}
\end{equation}
where
\begin{equation}
    Z = \sum_{z=0}^{2^{M N}-1} \omega^z \proj{z} \text{ with } \omega = \mathrm{e}^{2\pi \mathrm{i}\cdot 2^{-(M+N)}}
\end{equation}
is the usual $2^{M+N}$-dimensional Pauli-$Z$ operator. Note that $Z$ is observable but $\tilde{Z}$ is an unobservable.
Likewise, the operators 
\begin{equation}
    \tilde{X}:=H \tilde{Z} H^\dagger,
    \text{ and }\tilde{Y}:=\mathrm{i} H  \tilde{Z},
\end{equation}
where $H$ is the Fourier transformation, are unobservables. 
Analogously to \cref{eq:Ztilde} we can define $\tilde{X}$ and $\tilde{Y}$.
Then $\tilde{X}$, $\tilde{Y}$, and $\tilde{Z}$ are Pauli operators on a $2^{M+N}$-dimensional qudit.
Therefore, there is a qudit space on which all Pauli-operators are unobservables.
Note however, that the canonical qudit shares the same state space and the corresponding Pauli operators are observable. 

The unitary $U$ in \cref{eq:Ztilde} maps the unobservable $\tilde{Z}$ to the observable $Z$.
Thus $U$, as any unitary that transforms an unobservable into an observable, cannot be realized as a physical time evolution. This recapitulates the result of~\cite{aaronson2020hardness}, without relying on a specific quantifier of complexity, such as the quantum circuit complexity employed there.
Let $\mathcal{T}$ be the set of physically realized time evolutions using only the finite resources (time, energy, space) available to the observer. 
$\mathcal{T}$ does not contain all unitary operators, as illustrated by $U$. 
This motivates us to define the \emph{time-evolution orbit} of a mixed or pure state as 
\begin{equation}
\begin{aligned}
    \mathcal{T}(H):=&\{U H U^\dagger \;|\; U\in \mathcal{T}\} \\
    \text{ and } \mathcal{T}(\ket{\psi}) =& \{U \ket{\psi} \;|\; U\in \mathcal{T}\},
\end{aligned}
\end{equation}
respectively. 

We associate states
\begin{equation}
\begin{aligned}
    &\rho = \sum_{z=0}^{2^{M+N}-1} p_z \proj{z} \text{ and } \tilde{\rho} = U \rho U^\dagger\\
\text{with }&\sum_{z=0}^{2^{M+N}-1} p_z = 1 \text{ and } p_i\neq p_j \;\forall i \neq j,
\end{aligned}
\end{equation}
with the observable $Z$ and the unobservable $\tilde{Z}$, respectively. Any unitary $U'$ with $U'\rho U'^\dagger = \tilde{\rho}$ maps the observable $Z$ onto the unobservable $\tilde{Z}$ and is therefore not a physically realizable time evolution. 
Indeed such $U'$ is equal to $U$ up to a unitary which is diagonal in the computational basis. 


The complexity to evolve $\rho$ into $\tilde{\rho}$ is smaller or equal to the minimum complexity to evolve any purification of $\rho$ into any purification of $\tilde{\rho}$, as the latter accomplishes the former as well. 
This implies that there are also two pure states $\ket{\psi_1}$ and $\ket{\psi_2}$ which are not connected by any unitary time evolution in $\mathcal{T}$, i.e. $\ket{\psi_2}\not\in \mathcal{T}(\ket{\psi_1})$.
By definition of $\mathcal{T}$, there is no time-efficient way to map $\ket{\psi_1}$ to $\ket{\psi_2}$, i.e. any unitary time evolution $W(t)$ with $W(t)\ket{\psi_1} = \ket{\psi_2}$ requires exponential time $t=\exp(\Omega(N))$. 
Again we think of an $N$ which is large enough such that $t$ is prohibitively large for the finite resources of any observer.
The slow evolution implies a weak coupling and the transition probability $|\bra{\psi_2}W(t_0)\ket{\psi_1}|^2=\exp(-\Omega(N))$ is small for sub-exponential time $t_0$.
Hence, any unitary $V$ in $\mathcal{T}$ fulfills 
\begin{equation}
    \left|\bra{\psi_1} V \ket{\psi_2}\right| = \exp(-\Omega(N)).\quad \forall V\in\mathcal{T} \label{eq:unitaryinT}
\end{equation}
Now, employing the von Neumann model, any measurement is implemented as a unitary interaction on the system and the measurement apparatus followed by a measurement of the state of the apparatus~\cite{VonNeumann1932}.
For a unitary measurement operator $A=\sum_{i=1}^d \lambda_i A_i$ which associates labels $\lambda_i\in\mathds{C}$ to $d$ projectors $A_i$, the corresponding interaction can be written as
\begin{equation}
    C_A = \sum_{i=1}^d \proj{i} \otimes \lambda_i A_i
\end{equation}
and the measurement apparatus is initialised in the $\ket{+}:=\frac{1}{\sqrt{d}}\sum_i \ket{i}$ state and finally measured in the Fourier basis.
The case of non-unitary measurement operators is handled by relabeling the measurement outcomes.
Note that $A$ is not necessarily Hermitian, but it is a normal operator and the spectral theorem holds by construction.
For an observable $A$ we calculate
\begin{equation}
\begin{aligned}
\left|\bra{\psi_1} A \ket{\psi_2}\right| \overset{\hphantom{\text{\cref{eq:unitaryinT}}}}{=}& d^2 \left|\bra{+}\bra{\psi_1} C_A \ket{+}\ket{\psi_2}\right|\\
\overset{\text{\cref{eq:unitaryinT}}}{=}& \exp(-\Omega(N)). \label{eq:unitaryofmeasurement}
\end{aligned}
\end{equation}
We remark that for $A=\mathds{1}$ this implies that the two states $\ket{\psi_1}$ and $\ket{\psi_2}$ are practically orthogonal.
\cref{eq:unitaryofmeasurement} implies that the expectation values of $A$ in the states
\begin{equation}
    \begin{aligned}
        \rho_{\mathrm{coh}} =& \frac{1}{2}(\ket{\psi_1}+\ket{\psi_2})(\bra{\psi_1}+\bra{\psi_2})\\
        \text{and }\rho_{\mathrm{incoh}} =& \frac{1}{2}(\ket{\psi_1}\bra{\psi_2}+\ket{\psi_2}\bra{\psi_2}),
    \end{aligned}
\end{equation}
are indistinguishable, as
\begin{equation}
\begin{aligned}
    &\tr (A \rho_{\mathrm{coh}})-\tr (A \rho_{\mathrm{incoh}})\\
    =&2\mathrm{Re}(\bra{\psi_1}A\ket{\psi_2})
    \overset{\cref{eq:unitaryofmeasurement}}{=} \exp(-\Omega(n)).
\end{aligned}
\end{equation}
The coherent superposition of two states which are not connected by time-evolution cannot be distinguished from the incoherent mixture of the two.

The coherence of the state $\rho_{\mathrm{coh}}$ cannot be observed and in this sense we arrived at the concept of decoherence. This vanishing of interference between different (sets of) states is a hallmark of the presence of a \ac{SSR}. \ac{SSR}s were introduced in 1952 by Wick, Wightman, and Wigner~\cite{wick1952intrinsic} to account for the fact that certain quantities, like charge or spin, seem exempt from the superposition principle. The inability to detect interference between the states $\ket{\psi_1}$ and $\ket{\psi_2}$ entails that any admissible observable $O$ must be diagonal in the basis spanned by these states. 

Two relevant consequences of the existence of an \ac{SSR} are, first, that a formally pure superposition corresponds to a mixture, as shown above, and second, that this decomposition is unique. This stands in contrast to ordinary mixed states, which do not have a unique extremal decomposition in general. In particular, this suggests that it is permissible to apply an ignorance interpretation to such a mixture, interpreting the coefficients of its convex decomposition as probabilities of finding the system in a given state. On this basis, \ac{SSR}s have been proposed to deliver a `wash-out' solution to the measurement problem, where interference vanishes to leave only classical probabilities (e.g. see~\cite{hepp1972quantum, bub1988micro, bub1988solve, araki1986continuous, landsman1995observation, svozil2026unitarity}).

It remains to be seen whether the `effective' \ac{SSR} introduced here can fulfill this role. Of note, this is not the first time that limitations of the observer have been appealed to with regard to \acp{SSR}: Landsman, e.g., appeals to the localized nature of the observer to justify the impossibility of detecting interference between sufficiently separated states~\cite{landsman1995observation}, while the program of \textit{einselection} (environmentally induced superselection) seeks to similarly define effective or \ac{FAPP} superselection rules based on environmental decoherence effects~\cite{zurek2003decoherence}.
Along a different route, Peres has argued that the complexity of performing certain measurements is the reason for the appearance of irreversibility~\cite{peres1980can}.
However, to the best of our knowledge, this is the first time that computational complexity has been investigated as a source of superselection.

Let us summarize the chain of arguments that have brought us here.
We started out by showing that the existence of computationally hard problems implies that there are measurements which cannot be performed.
More concretely, we showed that the sandwich-POVM $\mathcal{N}$ would allow us to efficiently solve 3SAT problems.
Thus, while these are perfectly valid measurements in the standard formalism of quantum mechanics, they cannot correspond to observables: rather, they represent unobservables.
For example, sufficiently large Hilbert spaces contain qudits where Pauli matrices correspond to unobservables.
Next we discussed the following consequences. 
There are unitary operators which are not physically realized as time evolutions in any system.
And there are even states which are not connected by any physical time evolution, leading to the concept of time-evolution orbitals.
Finally, we saw that no coherence can be observed between two states if one cannot time-evolve one into the other.
As discussed above this implies that we can switch from a quantum description of the state of a system to classical probabilities and in this sense classicality emerges.
In this paper we only showed the existence of this mechanism. 
We do not expect this mechanism to be the only reason for a classical world,
and a future quantitative analysis of unobservables may estimate how large its contribution is:
Are many quantum phenomena not visible in macroscopic systems, because observing them is too hard in a complexity theoretic sense?

\begin{acknowledgments}
This work was funded by the Quantum Computing Initiative of DLR via projects ALQU and R-QIP.
\end{acknowledgments}
%

\clearpage
\appendix
\onecolumngrid
\section{Undoing a Measurement with a Four-Qubit Code}
\label{sec:fourqubitcode}

In the following we describe a code that can undo a measurement with the smallest number of qubits possible.
The code $q_4$ is a \ac{CSS} code~\cite{CalderbankShor96, Steane1996} that encodes two logical qubits into four physical ones, with a code distance of two. Hence, in the usual notation, it is a [[4,2,2]] code. 
This code can detect a single-qubit error. If the position of the error is known, then it allows to correct it.
Hence it can correct a single erasure.
The stabilizer is generated by the operators
\begin{equation}
    \begin{aligned}
        g_1 =& X_1 X_2 X_3 X_4\\
        g_2 =& Z_1 Z_2 Z_3 Z_4
    \end{aligned}
\end{equation}
and we choose the logical operators
\begin{equation}
    \begin{aligned}
        \overline{X}_1 =& X_2 X_3,\\
        \overline{X}_2 =& X_2 X_4,\\
        \overline{Z}_1 =& Z_2 Z_4,\\
        \text{and } \overline{Z}_2 =& Z_2 Z_3.
    \end{aligned}
\end{equation}
The logical basis states are
\begin{equation}
    \begin{aligned}
        \ket{\overline{00}} =& \frac{1}{\sqrt{2}}\left(\ket{0000}+\ket{1111}\right),\\
        \ket{\overline{10}} =& \frac{1}{\sqrt{2}}\left(\ket{0110}+\ket{1001}\right),\\
        \ket{\overline{01}} =& \frac{1}{\sqrt{2}}\left(\ket{0101}+\ket{1010}\right),\\
        \text{and }\ket{\overline{11}} =& \frac{1}{\sqrt{2}}\left(\ket{0011}+\ket{1100}\right).
    \end{aligned}
\end{equation}
Performing a measurement on the first qubit (and ignoring the outcome) is equivalent to replacing it by a completely mixed state.
Mathematically, for any code word $\ket{\overline{\psi}}$
\begin{equation}
    \frac{\mathds{1}}{2} \otimes \tr_1 \proj{\overline{\psi}} = \frac{1}{4}\sum_{e\in \{\mathds{1},X_1,Y_1,Z_1\}} e\proj{\overline{\psi}}e^\dagger,
\end{equation}
which we interpret as a single qubit error $\mathds{1}$, $X$, $Y$, or $Z$ occurring with probability $\frac{1}{4}$ each.
This error can be corrected by performing a syndrome measurement,
and correcting it by applying an operation that anti-commutes with the corresponding stabilizer generators, see \cref{tab:422corrections}.
\begin{table}[tbhp]
    \caption{Possible choices of correction operations depending on the syndrome.}%
    \label{tab:422corrections}%
    \centering%
    \begin{tabular}{c|c}%
        Syndrome & Correction \\%
        \hline%
        00 & $\mathds{1}$\\%
        01 & $X_1$\\%
        10 & $Z_1$\\%
        11 & $Y_1$%
    \end{tabular}%
\end{table}%
Let us take a closer look at how this works.
Assume we started out in the state 
\begin{equation}
    \ket{\overline{\psi}} = \alpha \ket{\overline{00}} + \beta \ket{\overline{10}} + \gamma \ket{\overline{01}} + \delta \ket{\overline{11}}
\end{equation}
and measured $\ket{0}$ on the first qubit.
The state collapsed into
\begin{equation}
    \begin{aligned}
        \ket{\psi'} =& \proj{0}_1 \ket{\overline{\psi}}\\
            \propto& \ket{0} \left(\alpha \ket{000} + \beta\ket{110}+\gamma \ket{101} + \delta\ket{011}\right).
    \end{aligned}
\end{equation}
We now do a syndrome measurement, where the outcome of $g_2$ is always $+1$, however, the outcome of $g_1$ is random.
We consider both cases:
\begin{enumerate}[nosep]
    \item The syndrome is $00$, so we apply no correction. The state is $\frac{1}{2}(\mathds{1}+g_1) \ket{\psi'} \propto \ket{\overline{\psi}}$.
    \item The syndrome is $10$, so we apply $Z_1$ and the state is
    $Z_1 \frac{1}{2}(\mathds{1}-g_1) \ket{\psi'} \propto \ket{\overline{\psi}}$.
\end{enumerate}
Either way, we restored the original state that described the system prior to the collapse. 
A similar analysis can be done for other measurement bases.
It is noteworthy that the entanglement is restored by the measurement, not by the correction, which only affects a single qubit.

As the measurement is reversible, the ``collapse'' does not erase any information.
It can therefore also be achieved with a unitary operation, which is another way of seeing that it can be undone, see \cref{fig:revmeasurementequalsunitary}.
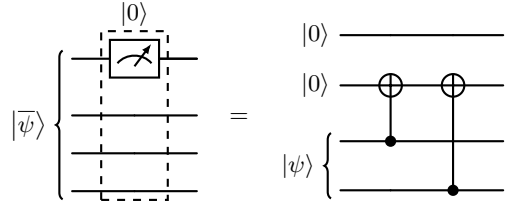
\begin{figure}[thbp]%
\centering%
\begin{quantikz}%
\lstick[4]{\ket{\overline{\psi}}} & \meter{}\gategroup[4,style={inner
sep=0pt,dashed}]{\ket{0}} & \qw\\%
 & & \\%
 & & \\%
 & & %
 \end{quantikz}%
\quad =\quad%
\begin{quantikz}%
\lstick{\ket{0}} & & & \\%
\lstick{\ket{0}}& \targ{}& \targ{} & \\%
\lstick[2]{\ket{\psi}} & \ctrl{-1}& & \\%
& & \ctrl{-2} &%
\end{quantikz}%
\caption{A unitary quantum circuit that implements the same operation as the reversible measurement, where only the $\ket{0}$-outcome is accepted.}\label{fig:revmeasurementequalsunitary}%
\end{figure}%
The reversible measurement could also be replaced by a measurement with post-selection.
However, the success probability then scales differently with the number of measurements.


\section{Detailed description of the algorithm}
\label{sec:algorithm}
Consider a propositional logic formula in \ac{CNF} with $N$ variables $v_1$, $v_2$, \dots, $v_N$ and $M$ clauses of exactly three different literals each.
We denote positive literals with $+v_i$ and negative literals (i.e. $\lnot v_i$) with $-v_i$. 
The formula can be expressed as
\begin{alignat}{2}
    F =& \bigwedge_{m=1}^M C_m,\\
\text{with clauses }
   C_m =& \bigvee_{j=1}^3  l_{mj}\quad m=1,2,...,M
\end{alignat}
and literals $l_{mj} = \sigma_{mj} v_{q_{mj}}$, where $\sigma_{mj}=\pm 1$ and $q_{mj}\in\{1,2,\dots,N\}$.
The task is to determine whether there is an assignment of truth values to the variables $v_1$, $v_2$, ..., $v_N$, such that $F$ evaluates to \true. 
If such an assignment exists, we call $F$ satisfiable.
The algorithm described below actually also returns an assignment of truth values to the variables if $F$ is satisfiable, but we will not need it.

We associate a qubit with each of the $N$ variables, such that $\ket{0}_i$ and $\ket{1}_i$ on qubit $i$ correspond to the values \false{} and \true{} of the variable $v_i$, respectively.
We denote the corresponding state space $\mathcal{H}_{2^N}$.

In the following we will define 
different mathematically valid measurements related to the given formula $F$ as defined above.

First we associate a projector $\Pi(F)$ with an arbitrary propositional logic formula $F$ on (a subset of) the $N$ variables such that it projects onto the space of fulfilling assignments of $F$.
We can write 
\begin{equation}
    \Pi(F) = \sum_{\substack{v_1,v_2,...,v_n=0\\F(v_1,v_2,...,v_N)=\true}}^1 \proj{v_1}\otimes \proj{v_2}\otimes ... \otimes \proj{v_N}.
\end{equation}
The following special cases will be useful in constructing the algorithm.
The projectors
\begin{equation}
\begin{aligned}
\Pi(C_m) =& \mathds{1}- \Pi(\lnot C_m)\\
    \text{and }\Pi(\lnot {C}_{m}) =& \prod_{j=1}^3 \Pi(\lnot l_{mj})
\end{aligned}
\end{equation}
with
\begin{equation}
    \Pi(\lnot l_{mj}) = \frac{\mathds{1} + \sigma_{mj} Z_{q_{m,j}}}{2}
\end{equation}
project onto the subspace of states fulfilling and violating the $m$-th clause in $F$, respectively. 
We denote the dichotomic \ac{PVM} formed by these projectors by
\begin{equation}
    \mathcal{M}(C_m) = \{\Pi(C_m),\; \Pi(\lnot C_m)\}. \label{eq:defmeasurementappendix}
\end{equation}
So $\mathcal{M}(C_m)$ answers the question whether the state of the $N$ qubits fulfills or violates the $m$-th clause of $F$.
As $\Pi(C_m)$ only involves three qubits, the measurement $\mathcal{M}(C_m)$ is easy to implement independently of $N$ and $M$, see \cref{fig:measurementforclause}.

If we would measure the `yes'-outcome for all clauses, we would know for sure that the formula $F$ would be satisfiable.
However, to tackle the undesired `no'-outcome, we introduce further measurements.

First we introduce a measurement that will be used below to ensure that the probabability of the `no' outcome will be not larger than $\frac{1}{2}$.
It is defined analogously to $\mathcal{M}(C_m)$ in \cref{eq:defmeasurementappendix} as
\begin{equation}
\mathcal{M}\left(\ubigvee_{i\in S} v_i\right) \label{eq:xormeasurement}   
\end{equation}
for some $S\subset \{1,2,..,N\}$, except that we use the \ac{XOR} in contrast to the ``normal or'' that appears in the clauses of $F$.
Also this measurement can be implemented efficiently, see \cref{fig:measurementforxorclause} for an example.

The sequential measurement of $\mathcal{M}$ for all clauses in $F$ with outcomes $x=(x_1, x_2, ..., x_m)$ implements the projection
\begin{equation}
    P_x(F) = \prod_{i} \left\{\begin{array}{cc}
    \Pi(C_i) & \text{ if $x_i=0$}\\
    \Pi(\lnot C_i) & \text{ else.}
    \end{array}\right.
\end{equation}
We can also interpret $x_i$ as the $i$-th digit of an $m$-digit binary number $x=0,1,2,...,2^m-1$.
Note that also this measurement can be implemented efficiently.
The last measurement we need is a simple measurement in the $X$-basis on every qubit with projectors
\begin{equation}
    Q_y = \bigotimes_{i=1}^N \frac{\mathds{1}+(-1)^{y_i} X}{2},
\end{equation}
where $y_i$ is the outcome on qubit $i$. Again we can interpret $y_i$ as the $i$-th digit of $y=0,1,2,...,2^N-1$
Note that $P_x(F)$ and $Q_y$ do not commute.

Now we gathered everything to finally define the 'sandwich'-measurement for a (\ac{CNF}-\ac{XOR}) formula $F=\bigwedge_{i=1}^M C_i$,
\begin{equation}
    \mathcal{N}(F) = \{P_x(F) Q_y P_x(F)| 0\leq x \leq 2^M-1, 0\leq y\leq 2^N-1\}, \label{eq:sandwichmeasurementappendix}
\end{equation}
where we label the outcomes with a multi-index to match the above definitions.
We verify that $\mathcal{N}(f)$ is a \ac{POVM}, hence a mathematically valid quantum measurement. As the $P_x(f)$ and the $Q_y$ are projectors, the elements of $\mathcal{N}(f)$ are positive semidefinite. Also the completeness condition
\begin{equation}
    \sum_{x,y} P_x(F) Q_y P_x(F)  = \sum_x P_x(F) = \mathds{1}
\end{equation}
holds. 
Even if $P_x(F)$ and $Q_y$ are easy to implement, and despite the compact form of \cref{eq:sandwichmeasurementappendix}, $\mathcal{N}$ may not be implementable for large $N$.
From the point of view of quantum computing this would not be surprising, it requires exponentially (in $N$) many gates to implement general operations on $N$ qubits in a quantum circuit.
But what if it so happens that $\mathcal{N}$ is realized natively in some quantum system?
We therefore seek a stronger argument than circuit complexity itself.
Indeed we will lead the following main assumption to a contradiction with the natural expectation that \ac{NP}-hard problems should not be `easily' solvable:
\begin{description}
\item[Assumption (OBS)] $\mathcal{N}(F)$ is an observable for a formula $F$ whose satisfiability problem is too hard to solve (with the finite resources available to any observer).
\end{description}

If for now we assume that (OBS) holds, then
we can use $\mathcal{N}$ to undo measurements of $\mathcal{M}$ in the following sense.
Let 
\begin{equation}
 \zeta(F)=\{v | F(v_1,v_2,...,v_N)=\true\}
\end{equation}
denote the set of solutions of the propositional formula $F$ and
\begin{equation}
    \chi(F) := |\zeta(F)|
\end{equation}
denote the number of solutions.
Now let $e$ and $F= e \land c$ be propositional formulas, which implies  $\zeta(F)\subset\zeta(e)$ and
\begin{equation}
    \Pi(F)\Pi(e) = \Pi(f).
\end{equation}
Furthermore we define the superposition of those satisfying assignments as
\begin{equation}
    \ket{\zeta(F)} = \sum_{v\in\zeta(F)} \ket{v_1,v_2,...,v_N}.
\end{equation}
Consider the scenario where we first measured $\mathcal{M}(e)$ on $\ket{+}$ with `yes'-outcome and then we measured $\mathcal{M}(c)$ with `no'-outcome, so the state is
\begin{equation}
\begin{aligned}
    \ket{\psi_f} =& \Pi(\lnot c)\Pi(e) \ket{+}\\
    \propto& \Pi(\lnot c) \ket{\zeta(e)}\\
    \propto& \ket{\zeta(e \land \lnot c)}.
\end{aligned}
\end{equation}
The measurement $\mathcal{N}(e)$ undoes the effect of $\Pi(\lnot c)$. 
In this context, the right $P_x(e)$ fixes the outcome $x=0$, $Q_y$ redistributes the probability amplitudes, while the left $P_x(e)$ projects back to the space of solutions of $e$.
The post-measurement state is
\begin{alignat}{2}
&P_0(e) Q_0 P_0(e) \Pi(\lnot c) \ket{\zeta(e)}\\
=&P_0(e) Q_0 \Pi(\lnot c)  P_0(e) \ket{\zeta(e)}\\
=&P_0(e) Q_0 \Pi(\lnot c)\ket{\zeta(e)}\\
\propto&P_0(e) Q_0 \ket{\zeta(e\land \lnot c)}\\
=&P_0(e) \proj{+} \ket{\zeta(e\land \lnot c)}\\
\propto&P_0(e) \ket{+}\\
\propto&\ket{\zeta(e)}.
\end{alignat}
So indeed the projection corresponding to the unwanted measurement outcome has been removed from the state.
Other values of $y$ lead to irrelevant sign flips in this state. 
Indeed as the phases are irrelevant in our hypothetical algorithm one could even work with classical mixed states, further indicating that the proposed algorithm is not illustrating a quantum advantage but a weakness in the formalism, at least when applied naively.
However, no other values for $x$ are measured: The probability for the previously obtained ${x'}=(0,0,\dots,0)$ is $1$, because
\begin{alignat}{2}
    &\sum_y\tr\left( P_{0}(e) Q_y P_{0}(e) \proj{\zeta(e\land\lnot c)}\right)\\
    =&\tr\left( P_{0}(e) \left(\sum_y Q_y\right) P_{0}(e) \proj{\zeta(e\land\lnot c)}\right)\\
    =&\tr\left( P_{0}(e)^2 \proj{\zeta(e\land\lnot c)}\right)\\
    =&\tr\left( P_{0}(e) \proj{\zeta(e\land\lnot c)}\right)\\
    =&\tr\left(P_{0}(e)   \frac{P_{0}(e\land\lnot c)\proj{+} P_0(e\land\lnot c)}{\tr P_0(e\land\lnot c)\proj{+} P_0(e\land\lnot c)}\right)\\
    =&\tr\left(\frac{P_0(e\land\lnot c)\proj{+} P_0(e\land\lnot c)}{\tr P_0(e\land\lnot c)\proj{+} P_0(e\land\lnot c)}\right) = 1.
\end{alignat}
Thus the already established projections $\Pi(c_i)$ with $i<\mu$ can be protected, while $Q_y$ erases the projection $\Pi(\lnot c_\mu)$.

We exploit assumption (OBS) and the described undoing of the measurement $\mathcal{M}$ in the following algorithm to find a solution of $F$. See also \cref{fig:flowchart} for an overview of the main idea.

However, there is a technical difficulty that requires some additional effort:
We cannot guarantee that the success probability of the measurement $\Pi(C_i)$ for a clause $C_i$ in the state $\ket{\zeta(e)}$,
\begin{equation}
    \tr(\Pi(C_i)\proj{\zeta(e)}) = \frac{\chi(e\land C_i)}{\chi(e)},
\end{equation}
is high enough in the sense that its inverse, the expected number of trials,  scales polynomially in $N$. 
Indeed it is easy to construct a counter example, see Appendix~\ref{sec:examplebadordering}.
To circumvent this, we ``dilute'' clauses to weaken their effect on the solution space, increasing the success probability for measuring the $\Pi(C_i)$ outcome.
By this we mean that we replace the formula $F$ by a randomly chosen 'diluted' formula $\mathcal{D}(F)$ with $M'=6M$ clauses on $N'=N+M$ variables.
More  details on the dilution are given in Appendix~\ref{app:dilution}.

The distinction between satisfiable and unsatisfiable instances is achieved as follows.
For satisfiable formulas $F$, the success probability of the \ac{XOR}-clauses increases until it reaches $1$.
In contrast, if $F$ is not satisfiable, then the success probability decreases to zero.
The two scenarios can be distinguished with high probability in polynomial time.
We give a detailed example of the algorithm in Appendix~\ref{sec:example}.

\section{Hard Formula for the hypothetical Algorithm without Dilution}
\label{sec:examplebadordering}
The formula with $N$ variables and $N$ clauses
\begin{equation}
\bigwedge_{i=2}^N (\lnot v_1 \lor v_i) \land v_1
\end{equation}
illustrates why dilution or a similar technique is crucial. 
The measurement probabilites are
\begin{equation}
    p_m = \left\{\begin{array}{cc}
        \frac{2^m + 1}{2^m + 2} & \text{ if }m<N \\
    \frac{1}{2^N + 1} & \text{ if }m=N. 
    \end{array}\right.
\end{equation}
The success probability of the last measurement scales exponentially bad in $N$. The first $N-1$ clauses create a large inbalance between  $v_1=\true$ and $v_1=\false$ with few and many terms in the superposition, respectively. 
 
The formula can be adapted to have three unique literals per clause, e.g.
\begin{equation}
\begin{aligned}
    &(v_1\lor v_2\lor \neg v_3)
    \land (v_1\lor \neg v_2\lor v_3)
    \land (v_1\lor \neg v_2\lor \neg v_3)
    \land(\neg v_1\lor v_2\lor \neg v_3)\\
    \land&(\neg v_1\lor \neg v_2\lor v_3)
    \land(\neg v_1\lor \neg v_2\lor \neg v_3)
    \land \bigwedge_{i=4}^N (\lnot v_1 \lor v_2 \lor v_i) \land (v_1\lor v_2\lor v_3)
    \end{aligned}
\end{equation}
without changing the exponential decay in the measurement probability for the last clause.

\section{Propositional formula dilution}
\label{app:dilution}
In this Appendix we define the dilution $\mathcal{D}(F)$ of a propositional logical formula $F$ in \ac{CNF}, prove two properties and give some interpretation.
\begin{definition}[Diluted formula]
For a given propositional formula $F$ in $N$ variables with $M$-clauses in \ac{CNF}, we define the dilution $\mathcal{D}(F)$ as the distribution over all formulas of the form
\begin{equation}
    \mathcal{D}(F) = \bigwedge_{m=1}^M C_m'\land \bigwedge_{m=1}^M X_m \land \bigwedge_{m=1}^M v_{N+m},
    \label{eq:defdilution}
\end{equation}
with
\begin{alignat}{2}
   C_{m}':=&\left(\bigvee_{j=1}^3 l_{mj}\lor \lnot v_{N+m}\right)\land \bigwedge_{j=1}^3 (\lnot l_{mj} \lor v_{N+m}),\label{eq:implicationcnf}\\
\text{and } X_i:=&\ubigvee_{m\in S_i} v_{N+m} \veebar \left\{\begin{array}{cc}
        \true &\text{if } |S_i| \text{ is even} \\
        \false & \text{else}
    \end{array}\right.
\end{alignat}
obtained by randomly choosing the index sets $S_i\subseteq\{1,2,\dots,M\}$, such that each $m\in\{1,2,\dots,M\}$ has probability $\frac{1}{2}$ to be in $S_i$. 
Each formula drawn from $\mathcal{D}(F)$ has $N+M$ variables and $6 M$ clauses.
\end{definition}
Note that for us the order of the clauses in $F$ and $\mathcal{D}(F)$ matter, as they indicate the order in which the measurements are carried out.
Note that the $M$ new variables $v_{N+i}$ with $i=1,2,\dots,M$ encode whether the corresponding clauses $C_{i}$ are satisfied, as they are equivalent to
\begin{equation}
   C_{m}':= (v_{N+m} \leftrightarrow C_m).
\end{equation}
We remark that the crucial trick of the $M$ random \ac{XOR} clauses, that only check the parity of the number of clauses satisfied is inspired by \emph{XORSample}~\cite{Gomes2006}.
Every random \ac{XOR} clause removes half of the assignments that violate the original formula, as we will see in the proof of (D2) below.

We now prove the two lemmata that are referred to in the main text as properties (D1) and (D2) of the dilution.
\begin{lemma}[D1]
For each $F'$ drawn from $\mathcal{D}(F)$, $F'$ is satisfiable if and only if $F$ is satisfiable.
\end{lemma}
\begin{proof}
\begin{description}
\item[$F$ satisfiable $\Rightarrow$ $F'$ satisfiable]
One can choose $v_{N+m}=\true$ and for the first $N$ variables use the satisfying assignment of $F$. Then all \ac{XOR}-Clauses are satisfied by construction. Also the $C_m'$ are fulfilled, as they encode $v_{N+m} \leftrightarrow C_m$. Thus this assignment evaluates $F'$ to \true.
\item[$F'$ is satisfiable $\Rightarrow$ $F$ is satisfiable]
There exists a solution of $F'$, and any solution of $F'$ fulfills $v_{N+m}=\true$. Because the clauses $C_m'$ encode $v_{N+m} \leftrightarrow C_m$, the assignment of values to the first $N$ variables in the solution of $F'$ also solves $F$.
\end{description}
\end{proof}
\begin{lemma}[D2]
Let $F$ be a \ac{CNF} formula with $N$ variables and $M$ clauses.
When consecutively measuring $\mathcal{M}(C_i')$ for $i=1,2,..,6 M$ and $C_i'$ the clauses of $F'=\mathcal{D}(F)$ the probability to measure $\Pi(C_i')$ is strictly greater than $\frac{1}{2}$ if and only if $F'$ is satisfiable.
\end{lemma}
\begin{proof}
We look at the three types of clauses in $\mathcal{D}(F)$ separately.
The four clauses in \cref{eq:implicationcnf} contain four different variables $v_{q_{m1}}$, $v_{q_{m2}}$, $v_{q_{m3}}$, and $v_{n+m}$, such that there are $2^4=16$ possible assignments to them.
They are listed in \cref{tab:excludedassignmentsfirstclauses}, where we also indicate which clause excludes them from the solution set, if any. As can be seen, the clauses I to IV exclude one, four, two, and one assignments, respectively.
Thus the success probabilities of the measurement $\mathcal{M}$ associated with these four clauses are
$\frac{15}{16}$, $\frac{11}{15}$, $\frac{9}{11}$, and $\frac{8}{9}$.
Note that the product of these probabilities is $\frac{1}{2}$, i.e. together they exclude half of the solutions, namely where $v_{N+m}\neq C_m$.
In terms of the original variables $v_{q_{m1}}$, $v_{q_{m2}}$, and $v_{q_{m3}}$, all assignments still appear with the same relative weight, such that the next set of four clauses ($m\rightarrow m+1$) have the same probabilities.
\begin{table}[th]
    \centering
    \begin{tabular}{c|c|c|c|c}
        $l_{m1}$ & $l_{m1}$ & $l_{m1}$ & $v_{N+m}$ & Excluded by clause\\
        \hline
        \cellcolor{mygreen} 0 &\cellcolor{mygreen} 0 &\cellcolor{mygreen} 0 & 0 & --\\
        0 & 0 & 0 & 1 & I\\
        0 & 0 & 1 & 0 & IV\\
        \cellcolor{mygreen}0 & \cellcolor{mygreen}0 & \cellcolor{mygreen}1 & 1 & --\\
        0 & 1 & 0 & 0 & III\\
        \cellcolor{mygreen}0 & \cellcolor{mygreen}1 & \cellcolor{mygreen}0 & 1 & --\\
        0 & 1 & 1 & 0 & III\\
        \cellcolor{mygreen}0 & \cellcolor{mygreen}1 & \cellcolor{mygreen}1 & 1 & --\\
        1 & 0 & 0 & 0 & II\\
        \cellcolor{mygreen}1 & \cellcolor{mygreen}0 & \cellcolor{mygreen}0 & 1 & --\\
        1 & 0 & 1 & 0 & II\\
        \cellcolor{mygreen}1 & \cellcolor{mygreen}0 & \cellcolor{mygreen}1 & 1 & --\\
        1 & 1 & 0 & 0 & II\\
        \cellcolor{mygreen}1 & \cellcolor{mygreen}1 & \cellcolor{mygreen}0 & 1 & --\\
        1 & 1 & 1 & 0 & II\\
        \cellcolor{mygreen}1 & \cellcolor{mygreen}1 & \cellcolor{mygreen}1 & 1 & --
    \end{tabular}
    \caption{The clause $v_{N+m} \leftrightarrow C_m$ leads to four clauses in CNF, see \cref{eq:implicationcnf}, labelled I, II, III, and IV. All possible assignments to the four involved variables are listed alongside the clause which they violate, if any. Surviving assignments to the original variables are highlighted (\textcolor{mygreen}{$\blacksquare$}), illustrating that their relative weights do not change after restricting to the solutions of these clauses.
    \label{tab:excludedassignmentsfirstclauses}}
\end{table}

We now prove the statement for the random \ac{XOR}-clauses.
The $i$-th \ac{XOR} clause indicates whether an even number of clauses is satisfied from some subset $S_i\subseteq\{1,2,\dots,M\}$, where each $m\in\{1,2,\dots,M\}$ has probability $\frac{1}{2}$ to be in $S_i$. 
The added constant to the \ac{XOR}-clause ensures that it is satisfied by the assignment of \true{} to all variables.
Consider an assignment $\sigma\in\{\true,\false\}^N$, that fulfills a set of clauses with indices
\begin{equation}
    T = \{m | C_m(\sigma) = \true \}.
\end{equation}
This assignment $\sigma$ fulfills the \ac{XOR}-clause $X_i$ if the number of clauses not satisfied by $T$ that are included in $S_i$ is even, i.e.
\begin{equation}
 |S_i\setminus T|\bmod 2 = 0.
\end{equation}
The probability for this to happen is
\begin{equation}
    P(|S_i\setminus T| \text{ is even}) = \left\{\begin{array}{cc}
         1 & \text{ if } |T|=M\\
         \frac{1}{2} & \text{else.} 
    \end{array}
    \right.
\end{equation}
The first case corresponds to $|S_i\setminus T| = 0$, which is always even. 
For the second case there is one index $m$ which is not in $T$, but it is in $S_i$ with probability $\frac{1}{2}$. Thus $|S_i\setminus T|$ is even and odd with equal probability.
Considering both cases we have that $P(|S_i\setminus T| \text{ is even})$ is equal to or larger than $\frac{1}{2}$ and strictly larger than $\frac{1}{2}$ when averaged over $T$, i.e. when averaged over $\sigma$.
The $M$ \ac{XOR} clauses gradually exclude $v_{N+m}=\false$ for all $m$ ($M$ variables and $M$ independent constraints).

Suppose that the set of solutions before adding $X_i$ is $\zeta(e)$. 
For every $\sigma\in \zeta(e)$ the probability that it violates $X_i$ is less than one half.
This implies that on average 
\begin{equation}
    \tr\left(\Pi(X_i) \proj{\zeta(e)}\right) = \frac{\zeta(e\land X_i)}{\zeta(e)} \geq \frac{1}{2}.
\end{equation}

Finally, we look at the single-literal clauses $v_{N+m}$. As the $M$ \ac{XOR} clauses are expected to be independent conditions, together they enforce $v_{N+m}=\true$ with high probability. Therefore the success probability of the measurement associated with the last $M$ clauses is one with high probability. These clauses are not included for their effect on the solution space, but for technical reasons, as they proved to be useful above.
\end{proof}

We remark that the effect of the dilution can also be achieved by using the original clauses directly in the \ac{XOR} clauses without the additional variables.
This allows to keep the number of variables equal to $N$, at the cost of slightly more complicated measurements (which are still efficiently implementable, though).

\section{Example for the hypothetical algorithm}
\label{sec:example}
We illustrate \cref{alg:solveFusingN} for the example formula with $M=10$ clauses in $N=5$ variables,
\begin{equation}
\begin{aligned}
    f=&\hphantom{{}\land{}} (\neg v_4\lor \neg v_2\lor v_3)\land (v_2\lor \neg v_4\lor \neg v_5)\land (v_4\lor \neg v_3\lor \neg v_1)\land (v_5\lor v_2\lor v_4)\\
    &\land (v_4\lor v_2\lor \neg v_5)\land (v_5\lor \neg v_3\lor \neg v_4)\land (v_1\lor v_2\lor v_3)\land (\neg v_1\lor v_5\lor v_3)\\
    &\land (\neg v_3\lor \neg v_5\lor \neg v_4)\land (v_4\lor \neg v_2\lor v_1).
\end{aligned}
\end{equation}
If we would just measure these clauses without dilution, the success probabilities would be
\begin{equation}
    (p_m)_m = \left(\frac{7}{8},\frac{6}{7},\frac{5}{6},\frac{17}{20},\frac{14}{17},\frac{5}{7},\frac{9}{10},\frac{7}{9},\frac{5}{7},\frac{1}{5}\right),
\end{equation}
which illustrates why the dilution (or a similar technique) is required for the lower bound on the success probability.
See also Appendix~\ref{sec:examplebadordering} for a more extreme example.
The dilution leads to larger formulas with $N'=N+M=15$ and $M'=6 M =60$.
A possible instance of the diluted formula is
\begin{equation}
\begin{aligned}
\mathcal{D}(f)=&\hphantom{{}\land{}}(\neg v_4\lor \neg v_2\lor v_3\lor \neg v_6)\land (v_4\lor v_6)\land (v_2\lor v_6)\land (\neg v_3\lor v_6)\\
&\land (v_2\lor \neg v_4\lor \neg v_5\lor \neg v_7)\land (\neg v_2\lor v_7)\land (v_4\lor v_7)\land (v_5\lor v_7)\\
&\land (v_4\lor \neg v_3\lor \neg v_1\lor \neg v_8)\land (\neg v_4\lor v_8)\land (v_3\lor v_8)\land (v_1\lor v_8)\\
&\land (v_5\lor v_2\lor v_4\lor \neg v_9)\land (\neg v_5\lor v_9)\land (\neg v_2\lor v_9)\land (\neg v_4\lor v_9)\\
&\land (v_4\lor v_2\lor \neg v_5\lor \neg v_{10})\land (\neg v_4\lor v_{10})\land (\neg v_2\lor v_{10})\land (v_5\lor v_{10})\\
&\land (v_5\lor \neg v_3\lor \neg v_4\lor \neg v_{11})\land (\neg v_5\lor v_{11})\land (v_3\lor v_{11})\land (v_4\lor v_{11})\\
&\land (v_1\lor v_2\lor v_3\lor \neg v_{12})\land (\neg v_1\lor v_{12})\land (\neg v_2\lor v_{12})\land (\neg v_3\lor v_{12})\\
&\land (\neg v_1\lor v_5\lor v_3\lor \neg v_{13})\land (v_1\lor v_{13})\land (\neg v_5\lor v_{13})\land (\neg v_3\lor v_{13})\\
&\land (\neg v_3\lor \neg v_5\lor \neg v_4\lor \neg v_{14})\land (v_3\lor v_{14})\land (v_5\lor v_{14})\land (v_4\lor v_{14})\\
&\land (v_4\lor \neg v_2\lor v_1\lor \neg v_{15})\land (\neg v_4\lor v_{15})\land (v_2\lor v_{15})\land (\neg v_1\lor v_{15})\\
&\land (v_{10}\veebar v_{11}\veebar v_{12}\veebar v_{13}\veebar v_{14}\veebar v_7\veebar v_8)\\
&\land \neg (v_{11}\veebar v_{14}\veebar v_8\veebar v_9)\\
&\land (v_{12}\veebar v_{13}\veebar v_{14}\veebar v_{15}\veebar v_8)\\
&\land \neg (v_{10}\veebar v_{11}\veebar v_8\veebar v_9)\\
&\land (v_{12}\veebar v_{13}\veebar v_{14}\veebar v_7\veebar v_8)\\
&\land (v_{11}\veebar v_{12}\veebar v_{15}\veebar v_7\veebar v_9)\\
&\land \neg (v_{10}\veebar v_{11}\veebar v_{12}\veebar v_9\veebar \neg v_6)\\
&\land \neg (v_{11}\veebar v_{12}\veebar v_{15}\veebar v_7\veebar v_8\veebar v_9)\\
&\land \neg (v_{12}\veebar v_{15}\veebar v_7\veebar v_8\veebar \neg v_6)\\
&\land (v_{14}\veebar v_{15}\veebar v_9\veebar \neg v_6)\\
&\land v_6\land v_7\land v_8\land v_9\land v_{10}\land v_{11}\land v_{12}\land v_{13}\land v_{14}\land v_{15}.
\end{aligned}
\label{eq:exampledilution}
\end{equation}
\cref{tab:example_solutions} shows which assignments are removed from the solution set by adding these \ac{XOR} clauses.
\begin{table}[htbp]
    \centering
\begin{tabular}{ccccc|cccccccccc}
$v_1$ & $v_2$ & $v_3$ & $v_4$ & $v_5$ & $v_6$ & $v_7$ & $v_8$ & $v_9$ & $v_{10}$ & $v_{11}$ & $v_{12}$ & $v_{13}$ & $v_{14}$ & $v_{15}$\\
\hline
\rowcolor{mygreen} 1 & 1 & 1 & 1 & 1 & 1 & 1 & 1 & 1 & 1 & 1 & 1 & 1 & 0 & 1 \\
\rowcolor{mygreen} 1 & 1 & 1 & 1 & 0 & 1 & 1 & 1 & 1 & 1 & 0 & 1 & 1 & 1 & 1 \\
\rowcolor{mygreen} 1 & 1 & 1 & 0 & 1 & 1 & 1 & 0 & 1 & 1 & 1 & 1 & 1 & 1 & 1 \\
\rowcolor{mygreen}1 & 1 & 1 & 0 & 0 & 1 & 1 & 0 & 1 & 1 & 1 & 1 & 1 & 1 & 1 \\
\rowcolor{mygrey} 1 & 1 & 0 & 1 & 1 & 0 & 1 & 1 & 1 & 1 & 1 & 1 & 1 & 1 & 1 \\
\rowcolor{mygreen} 1 & 1 & 0 & 1 & 0 & 0 & 1 & 1 & 1 & 1 & 1 & 1 & 0 & 1 & 1 \\
 1 & 1 & 0 & 0 & 1 & 1 & 1 & 1 & 1 & 1 & 1 & 1 & 1 & 1 & 1 \\
\rowcolor{mygreen} 1 & 1 & 0 & 0 & 0 & 1 & 1 & 1 & 1 & 1 & 1 & 1 & 0 & 1 & 1 \\
\rowcolor{myblue} 1 & 0 & 1 & 1 & 1 & 1 & 0 & 1 & 1 & 1 & 1 & 1 & 1 & 0 & 1 \\
\rowcolor{mygreen} 1 & 0 & 1 & 1 & 0 & 1 & 1 & 1 & 1 & 1 & 0 & 1 & 1 & 1 & 1 \\
\rowcolor{myblue} 1 & 0 & 1 & 0 & 1 & 1 & 1 & 0 & 1 & 0 & 1 & 1 & 1 & 1 & 1 \\
\rowcolor{mygreen} 1 & 0 & 1 & 0 & 0 & 1 & 1 & 0 & 0 & 1 & 1 & 1 & 1 & 1 & 1 \\
\rowcolor{mygreen} 1 & 0 & 0 & 1 & 1 & 1 & 0 & 1 & 1 & 1 & 1 & 1 & 1 & 1 & 1 \\
\rowcolor{mygreen} 1 & 0 & 0 & 1 & 0 & 1 & 1 & 1 & 1 & 1 & 1 & 1 & 0 & 1 & 1 \\
\rowcolor{mygreen} 1 & 0 & 0 & 0 & 1 & 1 & 1 & 1 & 1 & 0 & 1 & 1 & 1 & 1 & 1 \\
\rowcolor{mygreen} 1 & 0 & 0 & 0 & 0 & 1 & 1 & 1 & 0 & 1 & 1 & 1 & 0 & 1 & 1 \\
\rowcolor{mygreen} 0 & 1 & 1 & 1 & 1 & 1 & 1 & 1 & 1 & 1 & 1 & 1 & 1 & 0 & 1 \\
\rowcolor{mygreen} 0 & 1 & 1 & 1 & 0 & 1 & 1 & 1 & 1 & 1 & 0 & 1 & 1 & 1 & 1 \\
\rowcolor{myorange} 0 & 1 & 1 & 0 & 1 & 1 & 1 & 1 & 1 & 1 & 1 & 1 & 1 & 1 & 0 \\
\rowcolor{myorange}  0 & 1 & 1 & 0 & 0 & 1 & 1 & 1 & 1 & 1 & 1 & 1 & 1 & 1 & 0 \\
\rowcolor{mygrey} 0 & 1 & 0 & 1 & 1 & 0 & 1 & 1 & 1 & 1 & 1 & 1 & 1 & 1 & 1 \\
\rowcolor{mygrey} 0 & 1 & 0 & 1 & 0 & 0 & 1 & 1 & 1 & 1 & 1 & 1 & 1 & 1 & 1 \\
\rowcolor{myorange}  0 & 1 & 0 & 0 & 1 & 1 & 1 & 1 & 1 & 1 & 1 & 1 & 1 & 1 & 0 \\
\rowcolor{myorange}  0 & 1 & 0 & 0 & 0 & 1 & 1 & 1 & 1 & 1 & 1 & 1 & 1 & 1 & 0 \\
\rowcolor{myblue} 0 & 0 & 1 & 1 & 1 & 1 & 0 & 1 & 1 & 1 & 1 & 1 & 1 & 0 & 1 \\
\rowcolor{mygreen} 0 & 0 & 1 & 1 & 0 & 1 & 1 & 1 & 1 & 1 & 0 & 1 & 1 & 1 & 1 \\
\rowcolor{mygreen} 0 & 0 & 1 & 0 & 1 & 1 & 1 & 1 & 1 & 0 & 1 & 1 & 1 & 1 & 1 \\
\rowcolor{myblue} 0 & 0 & 1 & 0 & 0 & 1 & 1 & 1 & 0 & 1 & 1 & 1 & 1 & 1 & 1 \\
\rowcolor{myorange}  0 & 0 & 0 & 1 & 1 & 1 & 0 & 1 & 1 & 1 & 1 & 0 & 1 & 1 & 1 \\
\rowcolor{mygreen} 0 & 0 & 0 & 1 & 0 & 1 & 1 & 1 & 1 & 1 & 1 & 0 & 1 & 1 & 1 \\
\rowcolor{myorange}  0 & 0 & 0 & 0 & 1 & 1 & 1 & 1 & 1 & 0 & 1 & 0 & 1 & 1 & 1 \\
\rowcolor{mygreen} 0 & 0 & 0 & 0 & 0 & 1 & 1 & 1 & 0 & 1 & 1 & 0 & 1 & 1 & 1 \\
\end{tabular}
    \caption{A list of all solutions to $\mathcal{D}(f)$ at the point where the \ac{XOR} clauses are added. Colors indicate which solutions are removed by $C_{41}'$ (\textcolor{mygreen}{$\blacksquare$}), $C_{42}'$ (\textcolor{myblue}{$\blacksquare$}), $C_{43}'$ (\textcolor{myorange}{$\blacksquare$}), and $C_{47}'$ (\textcolor{mygrey}{$\blacksquare$}). Finally the white row shows the unique solution of $f$. Here we used the clauses shown in \cref{eq:exampledilution}, which illustrate possible values for the random \ac{XOR} clauses.}
    \label{tab:example_solutions}
\end{table}
We numerically estimate the average success probabilities for the diluted formula as
\begin{equation}
\begin{aligned}
    (p_m')_m =& (0.9375, 0.733333, 0.818182, 0.888889,\\
    &0.9375, 0.733333, 0.818182, 0.888889,\\
&0.9375, 0.733333, 0.818182, 0.888889,\\
&0.9375, 0.733333, 0.818182, 0.888889,\\
&0.9375, 0.733333, 0.818182, 0.888889,\\
&0.9375, 0.733333, 0.818182, 0.888889,\\
&0.9375, 0.733333, 0.818182, 0.888889, \\
&0.9375, 0.733333, 0.818182, 0.888889,\\
&0.9375, 0.733333, 0.818182, 0.888889,\\
&0.9375, 0.733333, 0.818182, 0.888889,\\
&0.515088, 0.53246, 0.570157, 0.653825, 0.765144,\\ &0.860789, 0.92324, 0.960468, 0.980077, 0.9899,\\
&0.998856, 0.998448, 0.998458, 0.998664, 0.998874,\\ &0.999162, 0.999452, 0.999427, 0.999433, 0.99928).
\end{aligned}
\end{equation}
This numerically confirms that the success probability for each measurement is larger than $\frac{1}{2}$.

We now show the action of $\mathcal{N}$. Let us assume that we already successfully measured until the 42\textsuperscript{nd} clause $C_{42}'$. 
Then $e=\bigwedge_{i=1}^{42} C_i'$. 
At this stage the state is
\begin{equation}
\begin{aligned}
    \ket{\psi_{42}}=&\frac{1}{\sqrt{10}}(\hphantom{{}+{}} 
    \ket{000011111010111}
    + \ket{000111011110111}
    + \ket{010001111111110}\\
    &\hphantom{\frac{1}{\sqrt{10}}( }+ \ket{010011111111110}
    + \ket{010100111111111}
    + \ket{010110111111111}\\
    &\hphantom{\frac{1}{\sqrt{10}}( }+ \ket{011001111111110}
    + \ket{011011111111110}
    + \ket{110011111111111}\\
    &\hphantom{\frac{1}{\sqrt{10}}( }+ \ket{110110111111111}).
\end{aligned}
\end{equation}
Assume that next we obtain the outcome $\Pi(\lnot C_{43}')$, projecting onto the state
\begin{equation}
\begin{aligned}
    \ket{\nu} =&\frac{1}{\sqrt{6}}(\hphantom{{}+{}} 
     \ket{000011111010111}
    +\ket{000111011110111}
    +\ket{010001111111110}\\
    &\hphantom{\frac{1}{\sqrt{6}}(}+\ket{010011111111110}
    +\ket{011001111111110}
    +\ket{011011111111110}).
\end{aligned}
\end{equation}
We lost the amplitude for the solution of $f$. But we recover by measuring $\mathcal{N}(e)$.
\begin{alignat}{2}
P_0(e) Q_y P_0(e) \ket{\nu} =& P_0(e) Q_y \ket{\nu}\\
=& P_0(e) \ket{\pm}^{\otimes N'}\\
=& \ket{\psi_{42}}. \quad\text{(up to phases)}
\end{alignat}
At this point we can try to measure $C_{43}'$ again. Note, however, since this is a random \ac{XOR}-clause, it will most likely take a different value.

\section{Robustness of the main results}
\label{app:robustness}
In the main text the results are formulated w.r.t. exact operations. 
Most importantly only the exact implementation of $\mathcal{N}(F)$ is considered.
Thus only the unobservability of this exact operation is proven.
However, in reality any measurement is only implemented approximately anyways, as only a finite precision is achievable in any experiment.
It is therefore important to look at the robustness of our main results when considering approximate implementations of the discussed operations.
In this appendix, we provide some comments on this topic, while we must defer a precise treatment to subsequent work. 



Let $L = R m$ be the maximal number of operations in the algorithm. 
We now describe why any measurement that is $\epsilon/L$ close to the sandwich-POVM allows to solve F in polynomial time within an error $\epsilon$.

\begin{lemma}[Robustness of noisy algorithms]\label{thm:CPTP_composition}
Let $U=U_L \circ U_{L-1} \circ ... \circ U_1$ and $U=V_L \circ V_{L-1} \circ ... \circ V_1$ be compositions of quantum operations (CPTP maps) which are similar in the sense that
\begin{equation}
    ||U_i - V_i||_\diamond < \frac{\epsilon}{L}\quad\forall i\in\{1,2,...,L\}.
\end{equation}
Then the expectation value of a projector $\Pi$ on any input state $\rho$ differs at most by $\epsilon$, i.e.
\begin{equation}
    |\tr \left( U(\rho) \Pi \right)-\tr \left( V(\rho) \Pi \right)| < \epsilon
\end{equation}
for all states $\rho$.
\end{lemma}
\begin{proof}
\begin{alignat}{2}
    |\tr \left( U(\rho) \Pi \right)-\tr \left( V(\rho) \Pi \right)|=&|\tr \left( (U(\rho)-V(\rho)) \Pi \right)|\\
    \leq& \tr \left|\left( (U(\rho)-V(\rho)) \Pi\right| \right)\\
    =& \left\lVert(U(\rho)-V(\rho)) \Pi\right\rVert_1\\
    \leq& \left\lVert(U(\rho)-V(\rho)) \right\rVert_1 \left\lVert\Pi\right\rVert_\infty\\
    =&\left\lVert(U(\rho)-V(\rho)) \right\rVert_1 \\
    \leq & \max_\rho \left\lVert(U(\rho)-V(\rho)) \right\rVert_1\\
    =&\left\lVert U-V \right\rVert_\diamond\\
    \leq&\sum_i \left\lVert U_i-V_i \right\rVert_\diamond\\
    <&L \frac{\epsilon}{L} = \epsilon
\end{alignat}
\end{proof}
\begin{theorem}[Robustness of unobservables]
Let $L$ be the maximal number of operations in the execution of the algorithm.
Any measurement $\mathcal{N}_{\mathrm{appr.}}$ with
\begin{equation}
    \lVert\mathcal{N}-\mathcal{N}_{\mathrm{appr.}}\rVert_{\diamond} < \frac{\epsilon}{L}
\end{equation}
allows to decide whether $F$ is satisfiable in polynomial time, as the success probability differs from the ideal one only by
\begin{equation}
    1-\tr\left(\rho_{\mathrm{final}} \Pi(F)\right) < \epsilon.
\end{equation}
\end{theorem}
\begin{proof}
Use \cref{thm:CPTP_composition} with $U_i$ and $V_i$ the ideal and noisy versions of the $i$-th operation ($\mathcal{M}$, which is assumed to be ideal, or $\mathcal{N}$), respectively.
\end{proof}

There is a ball around the ideal sandwich-POVM of approximate sandwich-POVMs which are also not observable. In this sense our result is robust.

This consideration translates to the consequences discussed in the second part of the main paper. Most importantly, any quantum operation close to the unitary channel $t(\rho)=U\rho U^{\dagger}$ cannot correspond to a physical quantum operation, because they would allow to implement $\mathcal{N}$ approximately, which is not possible as discussed above. 
This also implies that the map from $\rho$ to $\tilde{\rho}$ cannot be implemented approximately.

\end{document}